\documentclass{fundam}

\publyear{22}
\papernumber{2149}
\volume{188}
\issue{4}

\finalVersionForARXIV
\usepackage{url} % takes care of hyperlinks, preferred over hyperref
\usepackage[ruled,lined]{algorithm2e}% provides Algorithm environment
\usepackage{graphicx}% allows for inclusion of graphic files (figures)
\usepackage{tikz}
\usetikzlibrary{automata, positioning, arrows}
\usepackage{cite}
\usepackage{lineno,hyperref}

\begin{document}

\title{Exact Wirelength of Embedding 3-Ary \boldmath$n$-Cubes \\
                       into Certain Cylinders and Trees}

\author{S. Rajeshwari\\
School of Advanced Sciences\\
Vellore Institute of Technology\\
 Chennai-600127,  India \\
raje14697@gmail.com
\and M. Rajesh\thanks{Address for correspondence: School of Computer Science and Engineering,
                        Vellore Institute of Technology, Chennai-600127,  India.  \newline \newline
                    \vspace*{-6mm}{\scriptsize{Received April 2022; \ accepted  April  2023.}}}
   \\
School of Computer Science and Engineering \\
Vellore Institute of Technology\\
Chennai-600127,  India\\
 rajesh.m@vit.ac.in
 }

 \maketitle

\runninghead{S. Rajeshwari and M. Rajesh}{Exact Wirelength of Embedding 3-Ary $n$-Cubes into certain Cylinders...}

\vspace*{5mm}
\begin{abstract}
  Graph embeddings play a significant role in the design and analysis
of parallel algorithms. It is a mapping of the topological structure of a guest graph $G$ into a host graph $H$, which is represented as a one-to-one mapping from the vertex set of the guest graph to the vertex set of the host graph. In multiprocessing systems, the interconnection networks enhance the efficient communication between the components in the system. Obtaining minimum wirelength in embedding problems is significant in the designing of networks and simulating one architecture by another. In this paper, we determine the wirelength of embedding 3-ary $n$-cubes into cylinders and certain trees.
\end{abstract}

\begin{keywords}
embedding, edge isoperimetric problem, congestion, wirelength, 3-ary $n$-cube
\end{keywords}

\section{Introduction}

A multiprocessor is a computer network designed for parallel processing. It has numerous nodes that communicate by passing messages through a network. The pattern of connecting the nodes in a multicomputer is described as an interconnection network. By embedding a guest graph into a host graph, an already formulated algorithm for the guest graph can be modified and used in the embedded host architecture \cite{rajasingh2004embedding}. Embedding and its implications are extensively studied in \cite{guu1997circular,DBLP:journals/amc/Yang09,DBLP:journals/dam/Bezrukov01,rajasingh2004embedding,DBLP:journals/dam/ManuelARR11}. Embedding has vast applications in the complex connection networks such as network compression \cite{feder1995clique}, visualization \cite{jungnickel2005graphs}, clustering \cite{white2005spectral},
link prediction \cite{liben2007link} and node classification \cite{bhagat2011node}. The efficiency of a graph embedding is determined by the optimal wirelength of the layout. The wirelength of a graph embedding originate from VLSI designs, data structures, networks that deal with parallel computing systems, biological models, structural engineering and so on \cite{xu2013topological}. The implementation of 100 billion transistors in a Chip Multi-processor (CMP) has become a reality as microprocessor technology advances into the nanoscale stage \cite{bjerregaard2006survey}. The chip architecture must consider how to efficiently use a high number of transistors. The complexity of chip design is also rising, making it increasingly challenging on improving the overall performance of the system by enhancing the performance of a single processing core. Due to the key benefits of network-on-chip (NoC) such as high integration, low power consumption, cheap cost and compact volume, it has become a widely used approach to designing very large-scale integration (VLSI) systems \cite{xiang2015multicast,benini2002networks}. Various NoC is analysed for effective communication in CMP \cite{chittamuru2018bignoc,gu2017mronoc,gu2017time,yang2018taonoc}. The topology structure must meet a few unique requirements for NoC, due to the area restriction on processors, interconnection network and overall wirelength of NoC has emerged as the most pressing problem of its effective communication. It is a secondary factor for NoC to take into account when calculating the cost of their interconnection networks. The cost of wiring for connectivity increases, with network complexity. Consequently, it is preferable to replace NoC with a conventional network for the complex networks serving as a counterpart, where the embedding problem becomes a key feature in analysing NoC performance.
The $k$-ary $n$-cube is a parallel architecture used in implementation and message latency \cite{bose1995lee,yang2015note,gu20143}. This architecture is the hypercube when $k=2$ and the torus when $k=3$. Hypercubes have been used in Ipsc/2 and Ipsc/860 and tori in J-Machine, Cray T3D and T3E  \cite{seitz1989submicron}. The topological properties of $k$-ary $n$-cubes have been explored in \cite{ashir1997embeddings, ashir2002fault}. Due to the advantageous topological properties of $3$-ary $n$-cube, $Q^{3}_{n}$ such as symmetricity, pancyclicity, short message latency and easy implementation it has been utilised to build multicomputers such as the Cray XT5, Blue Gene/L supercomputers \cite{bauer2009scalable} and CamCube \cite{abu2010symbiotic} systems. Embedding problem on $3$-ary $n$-cubes is extensively studied on paths, cycles with faulty nodes and links \cite{dong2010embedding,lv2018hamiltonian}. Further $3$-ary $n$-cubes have been embedded into paths, grids \cite{fan2019optimally} and 3D Torus \cite{fan2022communication}. Fan et al. \cite{fan2020reconfigurable} had studied the fault tolerance of $3$-ary $n$-cubes and embedding of the same into torus NoC. In this paper, the optimal wirelength is computed for embedding 3-ary $n$-cubes into certain cylinders and certain trees such as caterpillars, firecracker graphs and banana trees, which enables the efficient communication of 3-ary $n$-cubes onto the above-mentioned network-on-chip.

\section{Preliminaries}

This section consists of the preliminary work required for our subsequent work.
\begin{definition}\cite{bezrukov2000edge}
The edge isoperimetric problem is to find a subset of vertices in a given graph that induces the maximum number of edges among all subsets with the same number of vertices. In otherwords, for a given $r$, $1 \leq r \leq |V_G|$, the problem is to find $I_G(r)= \max_{A\subseteq V, |A|=r}|I_G(A)|$, where $I_G(A)=\{(u, v)\in E: u, v \in A\}$.
\end{definition}

\begin{definition}\cite{bezrukov1998}
Embedding of graph $G$ into graph $H$ is a one-to-one mapping $f$ : $V(G)\rightarrow V(H)$ such that $f$ induces a one-to-one mapping $P_f$ : $E(G) \rightarrow \{ P_f(u,v) : P_f(u,v) $ is a path in $H$ between $f(u)$ and $f(v)$, for every edge $(u,v)$ in $G$\}.
\end{definition}

\begin{definition}\cite{bezrukov1998}
For an edge $e\in E(H)$, let $c_{f}(e)$ denote the number of
edges $(u,v)$ of $G$ such that $e$ is in the path $P_{f} (u,v)$ between
vertices $f(u)$ and $f(v)$ in $H$. The wirelength of an embedding
$f$ of $G$ into $H$ is given by $WL_{f} (G,H) =\sum_{e \in E(H)}
c_{f}(e)$. The wirelength of embedding $G$
into $H$ is defined as
$WL(G,H) = min\{WL_{f} (G,H) : f$ is an embedding
from $G$ to $H\}$.
%The edge congestion of an embedding $f$ of $G$ into $H$ is given by,\[EC_{f} (G,H) = max\{EC_{f}(e)|e\in E(H)\}\] Then, the edge congestion of $G$ into $H$ is defined as $EC(G,H) = min\{EC_{f} (G,H)|f$ is an embedding from $G$ to $H\}$.
\end{definition}

\begin{remark}
For any set $S$ of edges of $H$, $c_{f}(S)=\sum_{e \in S} c_{f}(e)$.
\end{remark}

\begin{remark}
$\sum_{v\in V(G_i)}deg_{G}(v)$ denotes the sum of degree of all vertices in $G_i$, where $deg_{G}(v)$ is the number of edges incident on a vertex $v$.
\end{remark}

\begin{lemma}(\cite{miller2015minimum}, Congestion Lemma)\label{congestion lemma}
Let $f$ be an embedding of an arbitrary graph $G$ into $H$. Let $S$ be an edge cut of $H$ such that the removal of edges of $S$ separates $H$ into two components $H_1$ and $H_2$ and let $G_1=f^{-1} (H_1 )$ and $G_2=f^{-1} (H_2 )$. Also $S$ satisfies the following conditions:
\begin{enumerate}
   \item For every edge  $(a,b)\in G_i, i= 1,2, P_f(a,b)$ has no edges in $S$.
   \item For every edge $(a,b)$ in G with  $a \in G_1$ and $b \in G_2$, $P_f(a,b)$ has exactly one edge in $S$.
   \item 	$G_1$ and $G_2$ are maximum subgraphs.
	\end{enumerate}
Then, $c_f(S)=\sum_{v\in V(G_1)}deg_{G}(v)-2|E(G_1)| = \sum_{v\in V(G_2)}deg_{G}(v)-2|E(G_2)|$ and $c_f(S)$ is minimum.
\end{lemma}

\begin{remark}
    In Lemma \ref{congestion lemma}, if $G$ is a regular graph then $G_1$ is a maximum subgraph of $G$ implies that $G_2$ is also a maximum subgraph of $G$.
\end{remark}

\begin{lemma}(\cite{DBLP:journals/ipl/ArockiarajMRR11}, $k$-Partition Lemma)
Let $f:G \rightarrow H$ be an embedding. Let $[kE(H)]$ denote a multiset of edges of $H$ with each edge in $H$ repeated exactly $k$ times. Let ${S_1,S_2,...,S_r}$  be a partition of $[kE(H)]$ such that each $S_i$ is an edge cut of $H$ satisfying the Congestion Lemma. Then \[WL_f(G,H)= \frac{1}{k} \sum\limits_{i=1}^r c_f(S_i).\]
\end{lemma}

\section{\boldmath $3$-Ary $n$-cube, $Q_n^3$}

\begin{definition}\cite{hsieh2007panconnectivity}
The 3-ary $n$-cube, $Q^{3}_{n}\ (n\geq1)$ is defined to be a graph on $3^n$ vertices, each of the form $x=(x_{n-1},x_{n-2},...,x_{0})$, where $0\leq x_{i}\leq 2$ for $0\leq i\leq n-1$. Two vertices are joined by an edge if and only if there exists $j$, $0\leq j\leq n-1$, such that $x_{j}=y_{j}\pm 1$ $(mod \ 3)$ and $x_{i}=y_{i}$, for every $i\in \{0,1,...,j-1,j+1,...,n-1\}$.
\end{definition}
It is also recursively defined as the cartesian product of $n$ cycles of order 3,
\[Q^{3}_{n}=C_{3} \otimes C_{3}\otimes...\otimes C_{3} (n \ \text{times}).\]
Thus,
\begin{equation*}
  Q_n^3=\begin{cases}
    C_3, & \text{if $n=1$}.\\
    C_{3}\otimes Q_{n-1}^3, & \text{otherwise}.
  \end{cases}
\end{equation*}

\noindent Each $Q_n^3$ contains three copies of $Q_{n-1}^3$ as subgraphs. Recursively each $Q_{n-1}^3$ has three copies of $Q_{n-2}^3$ as subgraphs. Thus we can partition $Q_n^3$ into 3 disjoint isomorphic copies $Q_{n-1}^3(0)$, $Q_{n-1}^3(1)$, $Q_{n-1}^3(2)$, where $Q_{n-1}^3(k)$, $\forall \ 0\leq k \leq2$ denotes the subgraph induced by the vertices $\{(x=x_{n-1},x_{n-2},...,x_{i},..,x_{0}) \in V(Q_n^3)|x_{i}=i\}$, for any $i=0,1,2$. Each $Q_{n-1}^3(k)$  is a convex set of $Q_n^3$. $Q_{n}^3$ has $k^{n-1}$ edges, having a perfect matching between $Q_{n-1}^3(k)$ and $Q_{n-1}^3(k+1)$, $\forall \ 0\leq k \leq2$. $Q_{n-1}^3(k)$ and $Q_{n-1}^3(k+1)$ are adjacent subcubes, and the edges between them are called `bridges'. The $n$ dimensional $Q_n^3$ is $2n$-regular\cite{hsieh2007panconnectivity}. See Figure 1.

\begin{figure}[!h]
\vspace*{2mm}
\centering
 \includegraphics[width=0.6\textwidth]{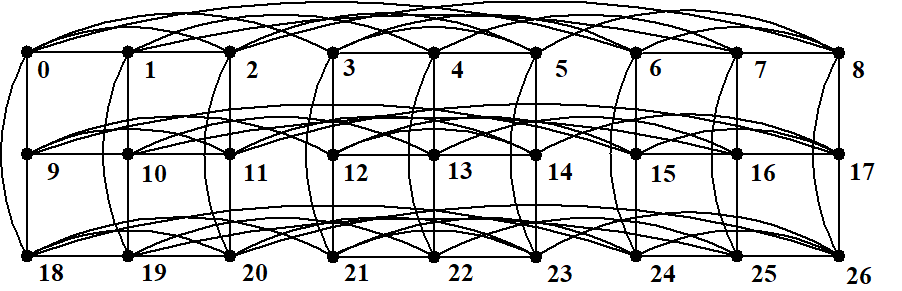}
 \caption{3-ary 3-cube, $Q_3^{3}$.}
 \label{Figure:1}\vspace*{-4mm}
\end{figure}

\begin{definition}\cite{DBLP:journals/tcs/BezrukovBK18}
The Lexicographic order on a set of $n$-tuples with integer entries is defined as follows: We say that $(x_1,...,x_n)$ is greater than $(y_1,...,y_n)$ if there exist an index $i$, $1\leq i \leq n$, such that $ x_j = y_j$ for $1 \leq j < i$  and $x_i > y_i$.
\end{definition}

\noindent Sergei et al.\cite{DBLP:journals/tcs/BezrukovBK18} has studied the edge isoperimetric problem for the torus $C_3 \times C_3$ which was solved in \cite{lindsey1964assignment, harper1967necessary} by introducing a new characteristic, called $\delta-${\it sequence} which is defined as follows: For a graph $G =(V, E)$ with $1\leq k \leq |V|$, we define
\begin{center}
 $\delta(k) = I(k)-I(k-1)$, with $\delta(1) = 0$,\\
 where $I(k)$ is the maximum number of edges induced by any $k$ vertices.\\
\end{center}
Further $\delta_G = (\delta(1), \delta(2)....,\delta(|V |))$ is called the $\delta-$sequence of $G$. The $\delta-${\it sequence} of $C_3 \times C_3$ is (0,1,2,1,2,3,2,3,4). This gives an optimal order for the maximum subgraph for $C_3 \times C_3$ by lexicographic ordering.

\begin{theorem}\cite{DBLP:journals/ejc/AhlswedeC97a}\label{AhlswedeC97a}
If the cartesian product of $G \times G$ is optimal with vertices of lexicographic ordering then it is optimal for $G^n$ for any $n \geq 3$.
 %If lexicographic order is optimal for $G \times G$ then it is optimal for $G^n$ for any $n \geq 3$.
 \end{theorem}

 \noindent
 The following corollary of Theorem 3.3 solves the edge isoperimetric problem in $Q_{n}^{3}$, $n\geq 2$.

\begin{corollary}\label{lex}
The Lexicographic ordering of vertices of $Q_{n}^{3}$, $n\geq 2$, is an optimal ordering for inducing maximum subgraphs in $Q_{n}^{3}$.
\end{corollary}

\begin{remark}
Let $lex_{k}=\{0,1,2,...,k-1\}$, $1\leq k\leq 3^{n}$ denote the first $k$ vertices in $Q_{n}^{3}$, $n\geq 2$ with lexicographic ordering.
\end{remark}

\begin{theorem}
If $G$ is a 3-ary $n$-cube, $Q_{n}^{3}$, $n\geq2$, then $I_G(k) = k_{1}3^{k_1}+(k_{2}+1)3^{k_2}+(k_{3}+2)3^{k_3}+...+(k_{r}+(r-1))3^{k_r}$, $k_i=0,1,2,...,n, \ 1\leq i\leq r$;
where $I_G(k)$ is the number of edges induced in any maximum subgraph on $k$ vertices and $k = 3^{k_1}+3^{k_2}+3^{k_3}+...+3^{k_r}$, $k_1\geq k_2\geq k_3\geq ....\geq k_r$.
\end{theorem}

\begin{proof}
Consider $Q_{n}^{3}(k)$ where  $k = 3^{k_1}+3^{k_2}+3^{k_3}+...+3^{k_r}$. $Q_{n}^{3}(k)$ contains $Q_{k_1}^{3}$, $Q_{k_2}^{3}$,...,$Q_{k_r}^{3}$ where $Q_{k_i}^{3}, \forall\  i>1 $ is adjacent with $Q_{k_{i-1}}^{3},...,Q_{k_2}^{3},Q_{k_1}^{3}$. There are $3^{k_i}$ edges between $Q_{k_i}^{3}$ and each of $Q_{k_j}^{3}$, $\forall \ j=1,2,...,i-1$. Thus there exist $(i-1)3^{k_i}$ edges between $Q_{k_i}^{3}$ and $Q_{k_j}^{3}$, $\forall \ j=1,2,...,i-1$. Further $Q_{k_i}^{3}$, $\forall \ i=1,2,3,...,r$ also has $k_i3^{k_i}$ edges in it. This implies that $Q_{k_i}^{3}$ contributes $(k_{i}3^{k_i}+(i-1)3^{k_i})=(k_{i}+(i-1))3^{k_i}$ edges to $I_{Q_{n}^{3}}(k)$. Hence the Lemma.
\end{proof}

\section{Embedding of \boldmath {$Q_n^3$} into cylinder \boldmath $C_{3}\times P_{3^{n-1}}$}

\begin{definition}\cite{DBLP:journals/dam/ManuelARR11}
Let $P_{\alpha}$ and $C_{\alpha}$ denote a path and cycle on $\alpha$ vertices respectively. The 2-dimensional grid is defined as $P_{\alpha_1}\times P_{\alpha_2}$, where $\alpha_i\geq2$ is an integer for each $i = 1, 2$. The cylinder $C_{\alpha_{1}}\times P_{\alpha_{2}}$, where $\alpha_{1}, \alpha_{2} \geq 3$ is a $P_{\alpha_1}\times P_{\alpha_2}$ grid with a wraparound edge in each column.
\end{definition}

\noindent \textbf{Lexicographic ordered embedding:} The lexicographic ordered embedding $lex:Q_{n}^{3}\rightarrow C_{3}\times P_{3^{n-1}}$ with labels 0 to $3^{n}-1$ is an assignment of labels to the vertices of $Q_{n}^{3}$ in lexicographic order and the vertices of $C_{3}\times P_{3^{n-1}}$ as follows: Vertices in $r^{th}$ column are labeled as $3(r-1)+0,\ 3(r-1)+1,\ 3(r-1)+2$ from top to bottom, where $r = 1,2,..,3^{n-1}$. \medskip\\
\textbf{Embedding Algorithm A:}\\ %of \boldmath{$Q_n^3$} into cylinder \boldmath{$C_{3}\times P_{3^{n-1}}$}
{\it Input:} The 3-ary $n$-cube, $Q_n^3$ and the cylinder $C_{3}\times P_{3^{n-1}}$ on $3^n$ vertices.\\
{\it Algorithm:} Lexicographic ordered embedding of $Q_n^3$ into $C_{3}\times P_{3^{n-1}}$.\\
{\it Output:} The embedding $lex$ of 3-ary $n$-cube, $Q_n^3$ into cylinder $C_{3}\times P_{3^{n-1}}$ on $3^n$ vertices is with minimum wirelength.

\medskip\noindent
{\bf Notation.} $C_{lex}^{i}= \{0,1,2,...,3i-1\}$, for $i=1,2,...,3^{n-1}-1$ denotes the first $i$ column vertices of $C_3\times P_{3^{n-1}}$ with vertices labeled as in Embedding Algorithm A. From Remark 3.5, it is clear that $C_{lex}^{i}=lex_{3i}$.
The following lemma is a consequence of Corollary 3.4.

\begin{lemma}\label{column}
$C_{lex}^{i}$ induces maximum subgraph in $Q_n^3$ for $i=1,2,..,3^{n-1}-1$.
\end{lemma}

\noindent \textbf{Notation.} $R_{lex}^{j}= \{j,3+j,...,3(3^{n-1})+j\}$, for $j=0,1,2$ denotes the $j^{th}$ row vertices of $C_3\times P_{3^{n-1}}$ with the lexicographic ordered embedding of $Q_n^3$ into $C_{3}\times P_{3^{n-1}}$.

\begin{lemma}\label{row}
$R_{lex}^{j}$ induces maximum subgraph in $Q_n^3$ for $j=0,1,2$.
\end{lemma}

\begin{proof}
From Lemma \ref{column}, we know that the lexicographic ordering columnwise induces a maximum subgraph. Hence to prove this lemma we have to show that the vertices in each row is isomorphic to subgraph induced by lexicographic ordering $0,1,2,...,3^{n-1}-1$. For $j=0,1,2$, define $\varphi^{j}:R_{lex}^{j}\rightarrow lex_{3^{n-1}}$ by $\varphi^{j}(3k+l)=3l+k+j$. If the $n$-tuple representation of integer $3k+l$ is $(\gamma_{1},\gamma_{2},...,\gamma_{n})$, then the $n$-tuple representation of integer $3l+k+j$ is $(\gamma_{2},\gamma_{3},...,\gamma_{n},\gamma_{1}+j)$. Thus if the $n$-tuple representation in two numbers $x$ and $y$ differ in exactly one bit, then it also holds good for $f(x)$ and $f(y)$. This implies that $(x,y)$ is an edge in $Q_n^3$ if and only if $(f(x),f(y))$ is an edge in $Q_n^3$. Thus $R_{lex}^{j}$ and $lex_{3^{n-1}}$ are isomorphic, which implies that $R_{lex}^{j}$ induces a maximum subgraph in $Q_{n}^{3}$.
\end{proof}

 \begin{theorem}
 The wirelength $WL(Q_n^3, C_3\times P_{3^{n-1}})$ is minimum for lexicographic ordered embedding $lex$ of $Q_n^3$ into $C_3\times P_{3^{n-1}}, n\geq2$.
 \end{theorem}

\begin{figure}[!h]
%\vspace*{2mm}
\centering
 \includegraphics[width=0.9\textwidth]{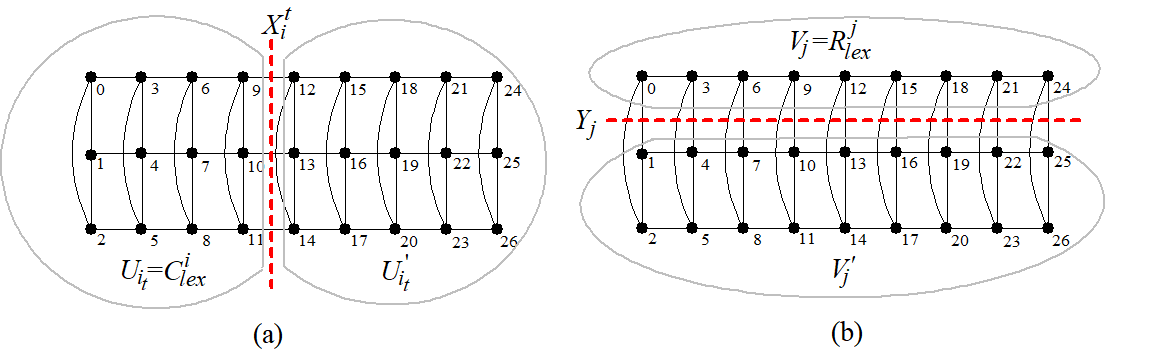}
    \caption{(a) Vertical edge cuts $X_{i}^{t}$, $1\leq i \leq 8$, $1\leq t \leq 2$ of Cylinder $C_3\times P_{9}$ with lexicographic ordering. (b)Horizontal edge cuts $Y_{j}$, $1\leq j \leq 2$ of Cylinder $C_3\times P_{9}$ with lexicographic ordering. }
    \label{Figure:2}%\vspace*{-2mm}
    \end{figure}

 \begin{proof}
 Consider the lexicographic embedding $lex:Q_n^3\rightarrow C_3\times P_{3^{n-1}}$ given in the Embedding Algorithm A. $X_{i}^{t}$, $i = 1,2,..,3^{n-1}-1$ and $t=1,2$, shown in Figure \ref{Figure:2}(a) is the vertical edge cut of the cylinder $C_3\times P_{3^{n-1}}$. Removal of $X_{i}^{t}$ disconnects $C_3\times P_{3^{n-1}}$ into two components $U_{i_{t}}$ and $U_{i_{t}}^{'}$, where $V(U_{i_{t}})=C_{lex}^{i}$. $Y_j$, $0\leq j \leq 2$ as shown in Figure \ref{Figure:2}(b) are the horizontal edge cuts of the cylinder $C_3\times P_{3^{n-1}}$. Thus $Y_j$ disconnects $C_3\times P_{3^{n-1}}$ into two components $V_j$ and $V_{j}^{'}$, where $V(V_{j})=R_{lex}^{j}$. See Figure \ref{Figure:2}(b). Let $S_{i_{t}}$ and $S_{i_{t}}^{'}$ be the preimages of $U_{i_{t}}$ and $U_{i_{t}}^{'}$ in $Q_n^3$ under lexicographic ordering respectively. The edge partition $X_{i}^{t}$ satisfies the first two conditions of the congestion lemma. To satisfy condition (iii) of the congestion lemma, it is enough to prove that the edges induced by the preimages $S_{i_{t}}$ and $S_{i_{t}}^{'}$ are maximum subgraphs. That is, congestion $c_{f}(X_{i}^{t})$ is minimum, where $S_{i_{t}}$ is the subgraph induced by the vertices of $C_{lex}^{i}$.
 By Lemma \ref{column}, $S_{i_{t}}$ is a maximum subgraph in $Q_n^3$. Hence by the Congestion Lemma $c_{f}(X_{i}^{t})$ is minimum for $i = 1,2,..,3^{n-1}-1$. Similarly, let $T_j$ and $T_{j^{'}}$ be the preimages of $V_j$ and $V_{j^{'}}$ in $Q_n^3$ under lexicographic ordering respectively. By Lemma \ref{row}, $T_j$ is a maximum subgraph induced by the vertices of $R_{lex}^{j}$. Hence by the Congestion Lemma $c_{f}(Y_j)$ is minimum for $j=0,1,2$. Partition Lemma consequently implies that $WL(Q_n^3,C_3\times P_{3^{n-1}})$ is minimum.
 \end{proof}

\begin{theorem}
The minimum wirelength of embedding $Q_n^3$ into $C_3\times P_{3^{n-1}}$ is given by \[WL(Q_n^3, C_3\times P_{3^{n-1}})=3^{n-1}\Big(2\big(3^{n-1}-1\big)+3 \Big).\]
\end{theorem}

\begin{proof}
By Congestion Lemma and 2-Partition Lemma,
 \begin{equation*}
     \begin{split}
      WL(Q_n^3, C_3\times P_{3^{n-1}} )&= \frac{1}{2}\bigg(\sum_{t=1}^{2}\sum_{i=1}^{(3^{n-1})-1} c_{lex}(X_{i}^{t})+\sum_{j=0}^{2} c_{lex}(Y_j)\bigg)\\
     &=\frac{1}{2}\bigg(4\Big(3^{n-1}\Big)\Big(3^{n-1}-1\Big)+6 \Big(3^{n-1}\Big)\bigg)\\
     &=3^{n-1}\Big(2\big(3^{n-1}-1\big)+3 \Big).
     \end{split}
 \end{equation*}

 \vspace*{-8mm}
\end{proof}

\section{Embedding of \boldmath {$Q_n^3$} into certain  trees}

 A tree is an acyclic connected graph. Trees are the most basic graph-theoretic models utilised in various domains, including automatic classification, information theory, data structure and analysis, artificial intelligence, algorithm design, operation research, combinatorial optimization, electrical network theory and network design \cite{xu2013topological}. We have embedded 3-ary $n$-cubes into certain trees such as caterpillar, firecracker graph and banana tree which are well known in the literature by satisfying the property of some graph variants \cite{wijaya,gallian2007survey,swaminathan2006super}. The research on caterpillars and their embeddings \cite{bezrukov1998embedding,manuel2011embedding} reveal that embedding problems are not simple. For instance, in \cite{haralambides1991bandwidth,monien1986bandwidth} the authors demonstrated the {\it NP}-completeness of determining the least dilation of embedding a caterpillar into chain. These predominant use of trees in networks motivated us to study the embedding of 3-ary $n$-cubes into certain trees mentioned above. In a tree traversal, labeling the vertices first time one visits is called preorder traversal.

\subsection{Wirelength of embedding \boldmath {$Q_n^3$} in caterpillar}

\begin{definition}\cite{DBLP:journals/dam/ManuelARR11}
 A {\it caterpillar} is a tree which will be a path if all its leaves are deleted. The path which is retained is called the backbone of the caterpillar.
 \end{definition}

\noindent  \textbf{Embedding Algorithm B:}\\% of \boldmath{$Q_n^3$} into Caterpillar
{\it Input:} The 3-ary $n$-cube, $Q_n^3$ and 2-regular caterpillar denoted by 2-CAT on $3^n$ vertices.\\
{\it Algorithm:} Label the vertices of 3-ary $n$-cube, $Q_n^3$ and caterpillar using lexicographic ordering and preorder traversal respectively.\\
{\it Output:} The embedding $lex$ of 3-ary $n$-cube, $Q_n^3$ into caterpillar on $3^n$ vertices is with minimum wirelength.
\begin{lemma}\label{cat edge cut}
The edge cuts $S_{i}$, $1 \leq i \leq 3^{n-1}-1$ and $T_{j}$, $1 \leq j \leq 2(3^{n-1})$ as shown in Figure \ref{Figure:3} induce maximum subgraphs in $Q_n^3$.
\end{lemma}

\begin{figure}[!h]
    \centering
   \includegraphics[width=0.84\textwidth]{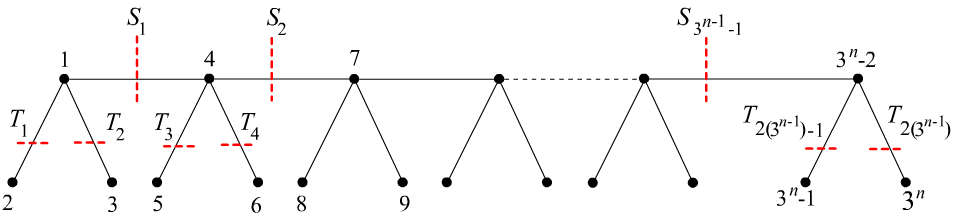}%\vspace{-1mm}
    \caption{Edge cuts of Caterpillar.}
    \label{Figure:3}\vspace*{-4mm}
\end{figure}

\begin{proof}
By Theorem \ref{AhlswedeC97a}, the lexicographic ordering of vertices of $Q_n^3$ gives the optimal order for inducing the maximum subgraph. The edge cut $S_i$ removes the edges in the backbone of the caterpillar, such that each $S_i$ disconnects it into two components of lexicographic ordering which induce a maximum subgraph in $Q_n^3$. The edge cut $T_j$ disconnects the caterpillar with exactly one vertex as one of the components. Hence $S_{i}, \forall \ i=1,2,...,3^{n-1}-1$ and $T_{j}, \forall \ j=1,2,...,2(3^{n-1})$ induce maximum subgraph in $Q_n^3$.
\end{proof}

\begin{lemma} The Embedding Algorithm B gives minimum wirelength of embedding $Q_n^3$ into 2-regular caterpillar.
\end{lemma}

\noindent {\bf Proof.} By Lemma \ref{cat edge cut} the edge cuts $S_{i}$ and $T_{j}$ satisfy conditions of the Congestion Lemma. Therefore $c_{f}(S_{i})$ and $c_{f}(T_{j})$ are minimum. Then the partition lemma implies that wirelength is minimum.

\begin{theorem}
The minimum wirelength of embedding $Q_n^3$ into caterpillar is given by \[WL(Q_n^3, \text{2-CAT})=2(3^{n-1})(3^{n-1}-1)+(4n)(3^{n-1}).\]
\end{theorem}

\begin{proof}
By Congestion Lemma and Partition Lemma,
\begin{align*}
  WL(Q_n^3, \text{2-CAT})&= \sum_{i=1}^{3^{n-1}-1} c_{f}(S_{i})+\sum_{j=1}^{2(3^{n-1})} c_{f}(T_{j})\\
     &= \sum_{i=1}^{3^{n-1}-1}((2n)(3i)-2|E(3i)|)+(2n)(2(3^{n-1}))\\
     &=2(3^{n-1})(3^{n-1}-1)+(4n)(3^{n-1}).
\end{align*}

\vspace*{-8mm}
\end{proof}

\subsection{Wirelength of embedding \boldmath {$Q_n^3$} in Firecracker graph}

\begin{definition}\cite{chen1997operations}
A {\it firecracker graph} $F_{n,k}$ is a graph obtained by the concatenation
of $n$, $k$-stars by linking one leaf from each.
 \end{definition}

\noindent  In what follows, we consider concatenation of $3^{n-1}$ number of 3-stars.

\medskip%\eject
\noindent \textbf{Embedding Algorithm C:}\\%of \boldmath{$Q_n^3$} into Firecracker graph
{\it Input:} The 3-ary $n$-cube, $Q_n^3$ and firecracker graph, $F_{3^{n-1},3}$ on $3^n$ vertices.\\
{\it Algorithm:} Label the vertices of 3-ary $n$-cube, $Q_n^3$ and firecracker graph, $F_{3^{n-1},3}$ using lexicographic ordering and preorder traversal respectively.\\
{\it Output:} The embedding $lex$ of 3-ary $n$-cube, $Q_n^3$ into firecracker graph, $F_{3^{n-1},3}$ on $3^n$ vertices is with minimum wirelength.

\begin{lemma}
The edge cut $S_{i}, \ \forall \ i=1,2,...,3^{n-1}-1$ of $F_{3^{n-1},3}$ as shown in Figure \ref{Figure:4} induces maximum subgraph in $Q_{n}^{3}$.
\end{lemma}

\begin{figure}[!h]
    \centering
   \includegraphics[width=0.65\textwidth]{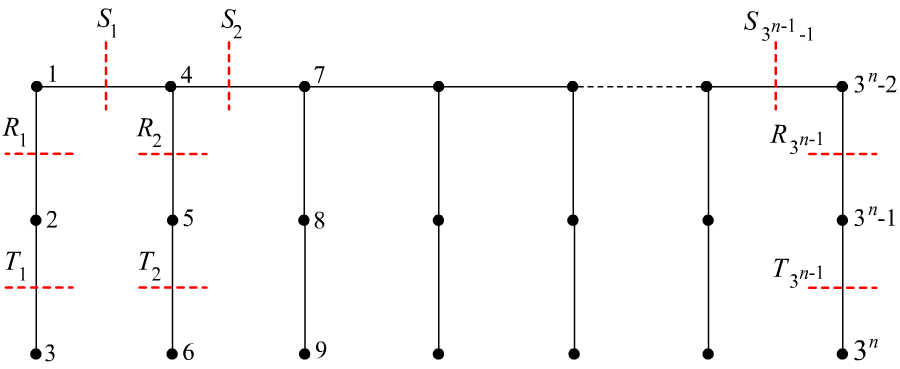}
    \caption{Edge cuts of $F_{3^{n-1},3}$.}
    \label{Figure:4}\vspace*{-3mm}
\end{figure}

\begin{proof}
The removal of edges in $S_{i}$, $1 \leq i \leq 3^{n-1}-1$ disconnects $F_{3^{n-1},3}$ into two components whose inverse images under $lex$ induce lexicographic ordering of the corresponding subgraphs of $Q_n^3$. This implies that the inverse images are maximum subgraphs of $Q_n^3$.
\end{proof}

\begin{lemma}
The edge cuts $R_{j}$ and $T_{k}, \forall \ j,k=1,2,...,3^{n-1}$ of $F_{3^{n-1},3}$ as shown in Figure \ref{Figure:4} induce maximum subgraph in $Q_n^3$.
\end{lemma}

\begin{proof}
The result is obvious as one of the components due to the cuts is either a singleton set or an edge.
\end{proof}
By Congestion Lemma and Partition Lemma, we arrive at the following result.

\begin{theorem}
Minimum wirelength is induced by the embedding algorithm of $Q_n^3$ into $F_{3^{n-1},3}$ on $3^n$ vertices.
\end{theorem}

\begin{theorem}
The minimum wirelength of embedding $Q_n^3$ into $F_{3^{n-1},3}$ is given by
\begin{equation*}
WL(Q_n^3, F_{3^{n-1},3})=2\big(3^{n-1}\big)\Big(\big(3^{n-1}-1\big)+\big( 2n-1\big)+n\Big).
\end{equation*}
%\[WL(C_3^n, F_{3^{n-1},3})=\sum_{i=1}^{3^{n-1}-1}((2n)(3i)-2|E(3i)|)+(4n-2)(3^{n-1})+(2n)(3^{n-1}).\]
\end{theorem}

\begin{proof}
By Congestion Lemma and Partition Lemma,
\begin{equation*}
    \begin{split}
     WL(Q_n^3, F_{3^{n-1},3})&= \sum_{i=1}^{(3^{n-1})-1} c_{f}(S_{i})+\sum_{j=1}^{3^{n-1}} c_{f}(R_{j})+\sum_{k=1}^{3^{n-1}} c_{f}(T_{k})\\
     %&= \sum_{i=1}^{(3^{n-1})-1}(r|V(G_{i})|-2|E(G_i)|)+ \sum_{j=1}^{3^{n-1}}(r|V(G_{j})|-2|E(G_j)|)\\&\qquad+\sum_{k=1}^{3^{n-1}}(r|V(G_{k})|-2|E(G_k)|)\\
     %&= \sum_{i=1}^{(3^{n-1})-1}((2n)(3i)-2|E(3i)|)+\sum_{j=1}^{3^{n-1}}((2n)(2j)-2(j))+\sum_{k=1}^{3^{n-1}}((2n)(k)-2(0))\\
     &= 2\big(3^{n-1}\times(3^{n-1}-1)\big)+(4n-2)(3^{n-1})+(2n)(3^{n-1})\\
     &=2\big(3^{n-1}\big)\Big(\big(3^{n-1}-1\big)+\big( 2n-1\big)+n\Big).
   %  &=\frac{1}{2}\bigg[2\Big((3^{\ceil{\frac{n}{2}}}-1)(6n-6)\Big)+\sum_{j=1}^{3}(2n(3^{n-1}j)-2|E(3^{n-1}j)|)\bigg].
  %   &= (3^{\floor{\frac{n}{2}}}-1)(r3^{\ceil{\frac{n}{2}}}-2|E(3^{\ceil{\frac{n}{2}}})|)+(3^{\ceil{\frac{n}{2}}}-1)(r3^{\floor{\frac{n}{2}}}-2|E(3^{\floor{\frac{n}{2}}})|)\\
    \end{split}
\end{equation*}

\vspace*{-8mm}
\end{proof}

\subsection{Wirelength of embedding \boldmath {$Q_n^3$} in banana tree}

%$B_{2,\big\lfloor\frac{3^n}{2}\big\rfloor}$}
\begin{definition}\cite{chen1997operations}
A banana tree $B_{n,k}$ is a graph formed by linking one leaf of each of $n$ copies of a $k$-star graph to a single root vertex that is different from all of the stars.
%Banana tree $B_{n,k}$ is defined as a graph obtained by connecting one leaf of each of $n$ copies of an $K$-star graph with single root vertex that is distinct from all the stars.
\end{definition}

 \noindent
\textbf{Embedding Algorithm D:}\\% of \boldmath{$Q_n^3$} into Banana Tree
{\it Input:} The 3-ary $n$-cube, $Q_n^3$ and banana tree, $B_{2,\big\lfloor\frac{3^n}{2}\big\rfloor}$ on $3^n$ vertices.\\
{\it Algorithm:} Label the vertices of 3-ary $n$-cube, $Q_n^3$ and banana tree, $B_{2,\big\lfloor\frac{3^n}{2}\big\rfloor}$ using lexicographic ordering and preorder traversal respectively.\\
{\it Output:} The embedding $lex$ of 3-ary $n$-cube, $Q_n^3$ into banana tree, $B_{2,\big\lfloor\frac{3^n}{2}\big\rfloor}$ on $3^n$ vertices is with minimum wirelength.

\begin{lemma}
The edge cuts $R_{j}$ and $T_{k}$, $\ \forall \ j,k=1,2$ of $B_{2,\big\lfloor\frac{3^n}{2}\big\rfloor}$ as shown in Figure \ref{Figure:5} induce maximum subgraphs in $Q_n^3$.
\end{lemma}

\begin{figure}[!h]
\vspace*{2mm}
    \centering
   \includegraphics[width=0.74\textwidth]{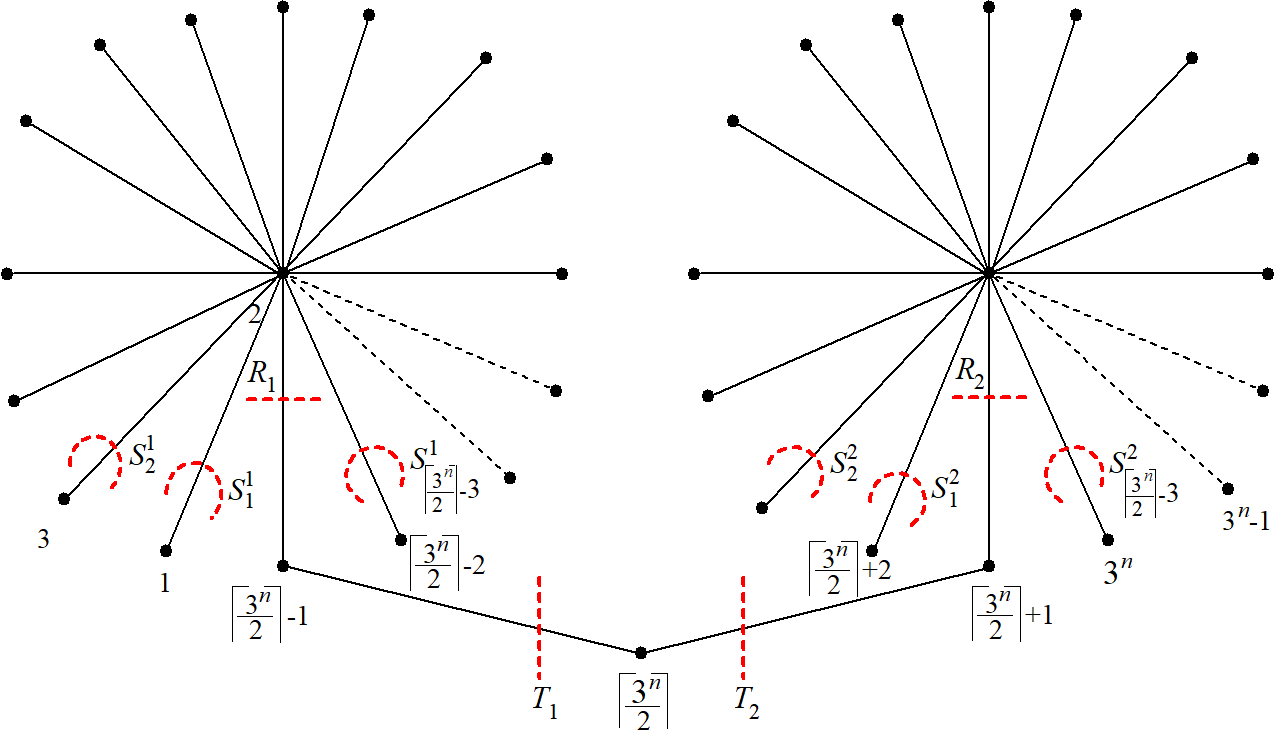}\vspace*{-1mm}
    \caption{Edge cuts of $B_{2,\big\lfloor\frac{3^n}{2}\big\rfloor}$.}
\label{Figure:5}\vspace*{-2mm}
\end{figure}

\begin{proof}
The removal of edges in $R_{j}$ and $T_{k}$, $\ \forall \ j,k=1,2$ disconnects $B_{2,\big\lfloor\frac{3^n}{2}\big\rfloor}$ into two components whose inverse images under $lex$ induce lexicographic ordering of the corresponding subgraphs of $Q_n^3$. This implies that the inverse images are maximum subgraphs of $Q_n^3$.
\end{proof}

\begin{lemma}
The edge cuts $S_{i}^1$ and $S_{i}^2$, $\forall \ i=1,2,...,\big\lceil\frac{3^n}{2}\big\rceil-3$ of $B_{2,\big\lfloor\frac{3^n}{2}\big\rfloor}$ as shown in Figure \ref{Figure:5} induce maximum subgraph in $Q_n^3$.
\end{lemma}

\begin{proof}
The result is obvious as one of the components due to the cuts is a singleton set.
\end{proof}
By Congestion Lemma and Partition Lemma, we arrive at the following result.

\begin{theorem}
Minimum wirelength is induced by the embedding algorithm of $Q_n^3$ into $B_{2,\big\lfloor\frac{3^n}{2}\big\rfloor}$ on $3^n$ vertices.
\end{theorem}

\begin{theorem}
The minimum wirelength of embedding $Q_n^3$ into $B_{2,\big\lfloor\frac{3^n}{2}\big\rfloor}$ is given by
\begin{equation*}
\begin{split}
 WL\Big(Q_n^3, B_{2,\big\lfloor\frac{3^n}{2}\big\rfloor}\Big)=4n\Big( \Big\lceil\frac{3^n}{2}\Big\rceil-3\Big)+4\Big( \Big\lceil\frac{3^n}{2}\Big\rceil-2\Big)+4\Big( \Big\lceil\frac{3^n}{2}\Big\rceil-1\Big).
\end{split}
\end{equation*}
\end{theorem}

\begin{proof}
By Congestion Lemma and Partition Lemma,
\begin{equation*}
\begin{split}
 WL\Big(Q_n^3, B_{2,\big\lfloor\frac{3^n}{2}\big\rfloor}\Big)&=\sum_{i=1}^{\big\lceil\frac{3^n}{2}\big\rceil-3}c_{f}(S_{i}^1)+\sum_{i=1}^{\big\lceil\frac{3^n}{2}\big\rceil-3}c_{f}(S_{i}^2)+\sum_{j=1}^{2} c_{f}(R_{j})+\sum_{k=1}^{2} c_{f}(T_{k})\\
 &=4n\Big( \Big\lceil\frac{3^n}{2}\Big\rceil-3\Big)+4\Big( \Big\lceil\frac{3^n}{2}\Big\rceil-2\Big)+4\Big( \Big\lceil\frac{3^n}{2}\Big\rceil-1\Big).
\end{split}
\end{equation*}

\vspace*{-8mm}
\end{proof}

\section{Conclusion}
The optimal wirelength of 3-ary $n$-cube into certain cylinders and trees such as caterpillars,
firecracker graphs and banana trees are determined in this paper.

\subsection*{Acknowledgments}
We extend our thanks to Dr. Indra Rajasingh,  Adjunct Professor, Saveetha School
of Engineering, Saveetha Institute of Medical and Technical Sciences, for her insightful suggestions. Also, we would like to express our gratitude to the anonymous reviewers for their thorough remarks, which allowed us to greatly enhance the paper.

%%%%%%%%%%%%%%%%%%%%%%%%%%%%%%%%%%%%%%%%%%%%%%%%%%%%%%%%%%%%%%%%%%%%%%


\begin{thebibliography}{10}
\providecommand{\url}[1]{\texttt{#1}}
\providecommand{\urlprefix}{URL }
\expandafter\ifx\csname urlstyle\endcsname\relax
  \providecommand{\doi}[1]{doi:\discretionary{}{}{}#1}\else
  \providecommand{\doi}{doi:\discretionary{}{}{}\begingroup
  \urlstyle{rm}\Url}\fi
\providecommand{\eprint}[2][]{\url{#2}}

\bibitem{rajasingh2004embedding}
Rajasingh I, William A, Quadras J, Manuel P.
\newblock Embedding of cycles and wheels into arbitrary trees.
\newblock \emph{Networks: An International Journal}, 2004.
\newblock \textbf{44}(3):173--178.

\bibitem{guu1997circular}
Guu CJ.
\newblock The circular wirelength problem for hypercubes.
\newblock University of California, Riverside, 1997.

\bibitem{DBLP:journals/amc/Yang09}
Yang M.
\newblock Path embedding in star graphs.
\newblock \emph{Appl. Math. Comput.}, 2009.
\newblock \textbf{207}(2):283--291.

\bibitem{DBLP:journals/dam/Bezrukov01}
Bezrukov SL.
\newblock Embedding complete trees into the hypercube.
\newblock \emph{Discret. Appl. Math.}, 2001.
\newblock \textbf{110}(2-3):101--119.
\newblock \doi{10.1016/S0166-218X(00)00256-0}.

\bibitem{DBLP:journals/dam/ManuelARR11}
Manuel PD, Arockiaraj M, Rajasingh I, Rajan B.
\newblock Embedding hypercubes into cylinders, snakes and caterpillars for
  minimizing wirelength.
\newblock \emph{Discret. Appl. Math.}, 2011.
\newblock \textbf{159}(17):2109--2116.
\newblock \doi{10.1016/j.dam.2011.07.003}.

\bibitem{feder1995clique}
Feder T, Motwani R.
\newblock Clique partitions, graph compression and speeding-up algorithms.
\newblock \emph{Journal of Computer and System Sciences}, 1995.
\newblock \textbf{51}(2):261--272.  doi:10.1006/jcss.1995.1065.

\bibitem{jungnickel2005graphs}
Jungnickel D.
\newblock Graphs, networks and algorithms.
\newblock Springer, 2005.  doi:10.1007/b138283.

\bibitem{white2005spectral}
White S, Smyth P.
\newblock A spectral clustering approach to finding communities in graphs.
\newblock In: Proceedings of the 2005 SIAM international conference on data
  mining. SIAM, 2005 pp. 274--285.

\bibitem{liben2007link}
Liben-Nowell D, Kleinberg J.
\newblock The link-prediction problem for social networks.
\newblock \emph{Journal of the American society for information science and
  technology}, 2007.
\newblock \textbf{58}(7):1019--1031.  doi:10.1002/asi.20591.

\bibitem{bhagat2011node}
Bhagat S, Cormode G, Muthukrishnan S.
\newblock Node classification in social networks.
\newblock In: Social network data analytics, pp. 115--148. Springer, 2011.
doi:10.1007/978-1-4419-8462-3\_5.

\bibitem{xu2013topological}
Xu J.
\newblock Topological structure and analysis of interconnection networks,
  volume~7.
\newblock Springer Science \& Business Media, 2013.

\bibitem{bjerregaard2006survey}
Bjerregaard T, Mahadevan S.
\newblock A survey of research and practices of network-on-chip.
\newblock \emph{ACM Computing Surveys (CSUR)}, 2006.
\newblock \textbf{38}(1):1--es. doi:10.1145/1132952.1132953.

\bibitem{xiang2015multicast}
Xiang D, Chakrabarty K, Fujiwara H.
\newblock Multicast-based testing and thermal-aware test scheduling for 3D ICs
  with a stacked network-on-chip.
\newblock \emph{IEEE Transactions on Computers}, 2015.
\newblock \textbf{65}(9):2767--2779.
doi:10.1109/TC.2015.2493548.

\bibitem{benini2002networks}
Benini L, De~Micheli G.
\newblock Networks on chips: A new SoC paradigm.
\newblock \emph{computer}, 2002.
\newblock \textbf{35}(1):70--78.
doi:10.1109/2.976921.

\bibitem{chittamuru2018bignoc}
Chittamuru SVR, Dang D, Pasricha S, Mahapatra R.
\newblock BiGNoC: Accelerating big data computing with application-specific
  photonic network-on-chip architectures.
\newblock \emph{IEEE Transactions on Parallel and Distributed Systems}, 2018.
\newblock \textbf{29}(11):2402--2415.
 doi:10.1109/TPDS.2018.2833876.

\bibitem{gu2017mronoc}
Gu H, Chen K, Yang Y, Chen Z, Zhang B.
\newblock MRONoC: A low latency and energy efficient on chip optical
  interconnect architecture.
\newblock \emph{IEEE Photonics Journal}, 2017.
\newblock \textbf{9}(1):1--12.
doi:10.1109/JPHOT.2017.2651586.

\bibitem{gu2017time}
Gu H, Wang Z, Zhang B, Yang Y, Wang K.
\newblock Time-division-multiplexing--wavelength-division-multiplexing-based
  architecture for ONoC.
\newblock \emph{Journal of Optical Communications and Networking}, 2017.
\newblock \textbf{9}(5):351--363.
 doi:10.1364/JOCN.9.000351.

\bibitem{yang2018taonoc}
Yang Y, Chen K, Gu H, Zhang B, Zhu L.
\newblock TAONoC: A regular passive optical network-on-chip architecture based
  on comb switches.
\newblock \emph{IEEE Transactions on Very Large Scale Integration (VLSI)
  Systems}, 2018.
\newblock \textbf{27}(4):954--963.
doi:10.1109/TVLSI.2018.2885141.

\bibitem{bose1995lee}
Bose B, Broeg B, Kwon Y, Ashir Y.
\newblock Lee distance and topological properties of k-ary n-cubes.
\newblock \emph{IEEE Transactions on Computers}, 1995.
\newblock \textbf{44}(8):1021--1030.  doi:10.1109/12.403718.

\bibitem{yang2015note}
Yang Y, Wang S.
\newblock A note on Hamiltonian paths and cycles with prescribed edges in the
  3-ary n-cube.
\newblock \emph{Information Sciences}, 2015.
\newblock \textbf{296}:42--45.  doi:10.1016/j.ins.2014.10.034.

\bibitem{gu20143}
Gu MM, Hao RX.
\newblock 3-extra connectivity of 3-ary n-cube networks.
\newblock \emph{Information Processing Letters}, 2014.
\newblock \textbf{114}(9):486--491.

\bibitem{seitz1989submicron}
Seitz CL.
\newblock Submicron systems architecture project: Semiannual technical report.
\newblock 1989.

\bibitem{ashir1997embeddings}
Ashir Y, Stewart IA.
\newblock Embeddings of cycles, meshes and tori in faulty k-ary n-cubes.
\newblock In: Proceedings 1997 International Conference on Parallel and
  Distributed Systems. IEEE, 1997 pp. 429--435.
doi:10.1109/ICPADS.1997.652583.

\bibitem{ashir2002fault}
Ashir YA, Stewart IA.
\newblock Fault-tolerant embeddings of Hamiltonian circuits in k-ary n-cubes.
\newblock \emph{SIAM Journal on Discrete Mathematics}, 2002.
\newblock \textbf{15}(3):317--328.  doi:10.1137/S08954801963111.

\bibitem{bauer2009scalable}
Bauer DW, Carothers CD.
\newblock Scalable RF propagation modeling on the IBM Blue
Gene/L and Cray XT5   supercomputers.
\newblock In: Proceedings of the 2009 Winter Simulation Conference (WSC). IEEE,
  2009 pp. 779--787.  doi:10.1109/WSC.2009.5429676.

\bibitem{abu2010symbiotic}
Abu-Libdeh H, Costa P, Rowstron A, O'Shea G, Donnelly A.
\newblock Symbiotic routing in future data centers.
\newblock In: Proceedings of the ACM SIGCOMM 2010 conference. 2010 pp. 51--62.
doi:10.1145/1851275.1851191.

\bibitem{dong2010embedding}
Dong Q, Yang X, Wang D.
\newblock Embedding paths and cycles in 3-ary n-cubes with faulty nodes and
  links.
\newblock \emph{Information Sciences}, 2010.
\newblock \textbf{180}(1):198--208.
doi:10.1016/j.ins.2009.09.002.

\bibitem{lv2018hamiltonian}
Lv Y, Lin CK, Fan J, Jia X.
\newblock Hamiltonian cycle and path embeddings in 3-ary n-cubes based on K1,
  3-structure faults.
\newblock \emph{Journal of Parallel and Distributed Computing}, 2018.
\newblock \textbf{120}:148--158.
doi:10.1016/j.jpdc.2018.06.007.

\bibitem{fan2019optimally}
Fan WB, Fan JX, Lin CK, Wang Y, Han YJ, Wang RC.
\newblock Optimally embedding 3-ary n-cubes into grids.
\newblock \emph{Journal of Computer Science and Technology}, 2019.
\newblock \textbf{34}(2):372--387.
 doi:10.1007/s11390-019-1893-0.

\bibitem{fan2022communication}
Fan W, Fan J, Zhang Y, Han Z, Chen G.
\newblock Communication and performance evaluation of 3-ary n-cubes onto
  network-on-chips.
\newblock \emph{Science China Information Sciences}, 2022.
\newblock \textbf{65}(7):1--3.
 doi:10.1007/s11432-019-2794-9.

\bibitem{fan2020reconfigurable}
Fan W, He J, Han Z, Li P, Wang R.
\newblock Reconfigurable Fault-tolerance mapping of ternary N-cubes onto chips.
\newblock \emph{Concurrency and Computation: Practice and Experience}, 2020.
\newblock \textbf{32}(11):e5659.
doi::10.1002/cpe.5659.

\bibitem{bezrukov2000edge}
Bezrukov SL, Das SK, Els{\"a}sser R.
\newblock An edge-isoperimetric problem for powers of the Petersen graph.
\newblock \emph{Annals of Combinatorics}, 2000.
\newblock \textbf{4}(2):153--169.
doi:10.1007/s000260050003.

\bibitem{bezrukov1998}
Bezrukov SL, Chavez JD, Harper LH, R{\"o}ttger M, Schroeder UP.
\newblock Embedding of hypercubes into grids.
\newblock In: Mathematical Foundations of Computer Science 1998: 23rd
  International Symposium, MFCS'98 Brno, Czech Republic, August 24--28, 1998
  Proceedings 23. Springer, 1998 pp. 693--701.
\newblock \doi{10.1007/BFb0055820}.

\bibitem{miller2015minimum}
Miller M, Rajan RS, Parthiban N, Rajasingh I.
\newblock Minimum linear arrangement of incomplete hypercubes.
\newblock \emph{The Computer Journal}, 2015.
\newblock \textbf{58}(2):331--337.
 doi:10.1093/comjnl/bxu031.

\bibitem{DBLP:journals/ipl/ArockiarajMRR11}
Arockiaraj M, Manuel PD, Rajasingh I, Rajan B.
\newblock Wirelength of 1-fault hamiltonian graphs into wheels and fans.
\newblock \emph{Inf. Process. Lett.}, 2011.
\newblock \textbf{111}(18):921--925.
\newblock \doi{10.1016/j.ipl.2011.06.011}.

\bibitem{hsieh2007panconnectivity}
Hsieh SY, Lin TJ, Huang HL.
\newblock Panconnectivity and edge-pancyclicity of 3-ary n-cubes.
\newblock \emph{The Journal of Supercomputing}, 2007.
\newblock \textbf{42}(2):225--233.
doi:10.1007/s11227-007-0133-5.

\bibitem{DBLP:journals/tcs/BezrukovBK18}
Bezrukov SL, Bulatovic P, Kuzmanovski N.
\newblock New infinite family of regular edge-isoperimetric graphs.
\newblock \emph{Theor. Comput. Sci.}, 2018.
\newblock \textbf{721}:42--53.
\newblock \doi{10.1016/j.tcs.2017.12.036}.

\bibitem{lindsey1964assignment}
Lindsey JH.
\newblock Assignment of numbers to vertices.
\newblock \emph{The American Mathematical Monthly}, 1964.
\newblock \textbf{71}(5):508--516.


\bibitem{harper1967necessary}
Harper L.
\newblock A necessary condition on minimal cube numberings.
\newblock \emph{Journal of Applied Probability}, 1967.
\newblock \textbf{4}(2):397--401.
 doi:10.2307/3212033.

\bibitem{DBLP:journals/ejc/AhlswedeC97a}
Ahlswede R, Cai N.
\newblock General Edge-isoperimetric Inequalities, Part {II:} a Local-Global
  Principle for Lexicographical Solutions.
\newblock \emph{Eur. J. Comb.}, 1997.
\newblock \textbf{18}(5):479--489.
\newblock \doi{10.1006/eujc.1996.0106}.

\bibitem{wijaya}
Wijaya R, Semaničová-Feňovčíková A, Ryan J, Kalinowski T.
\newblock H-supermagic labelings for firecrackers, banana trees and flowers.
\newblock \emph{Australasian Journal of Combinatorics}, 2017.
\newblock \textbf{69}:442--451.

\bibitem{gallian2007survey}
Gallian J.
\newblock A Survey: A dynamic survey on graph labeling.
\newblock \emph{Electron. J. Combin}, 2007.

\bibitem{swaminathan2006super}
Swaminathan V, Jeyanthi P.
\newblock Super edge-magic strength of fire crackers, banana trees and
  unicyclic graphs.
\newblock \emph{Discrete mathematics}, 2006.
\newblock \textbf{306}(14):1624--1636.
 doi:10.1016/j.disc.2005.06.038.

\bibitem{bezrukov1998embedding}
Bezrukov S, Monien B, Unger W, Wechsung G.
\newblock Embedding ladders and caterpillars into the hypercube.
\newblock \emph{Discrete Applied Mathematics}, 1998.
\newblock \textbf{83}(1-3):21--29.
 doi:10.1016/S0166-218X(97)00101-7.

\bibitem{manuel2011embedding}
Manuel P, Arockiaraj M, Rajasingh I, Rajan B.
\newblock Embedding hypercubes into cylinders, snakes and caterpillars for
  minimizing wirelength.
\newblock \emph{Discrete Applied Mathematics}, 2011.
\newblock \textbf{159}(17):2109--2116.
 doi:10.1016/j.dam.2011.07.003.

\bibitem{haralambides1991bandwidth}
Haralambides J, Makedon F, Monien B.
\newblock Bandwidth minimization: an approximation algorithm for caterpillars.
\newblock \emph{Mathematical Systems Theory}, 1991.
\newblock \textbf{24}(1):169--177.
 doi:10.1007/BF02090396.

\bibitem{monien1986bandwidth}
Monien B.
\newblock The bandwidth minimization problem for caterpillars with hair length
  3 is NP-complete.
\newblock \emph{SIAM Journal on Algebraic Discrete Methods}, 1986.
\newblock \textbf{7}(4):505--512.

\bibitem{chen1997operations}
Chen WC, Lu HI, Yeh YN.
\newblock Operations of interlaced trees and graceful trees.
\newblock \emph{Southeast Asian Bull. Math}, 1997.
\newblock \textbf{21}(4):337--348.

\bibitem{wang2016adaptive}
Wang Z, Gu H, Yang Y, Zhang H, Chen Y.
\newblock An adaptive partition-based multicast routing scheme for mesh-based
  networks-on-chip.
\newblock \emph{Computers \& Electrical Engineering}, 2016.
\newblock \textbf{51}:235--251.
doi:10.1016/j.compeleceng.2016.01.021.

\bibitem{Sugumaran2018DominationNO}
Sugumaran AK, Jayachandran E.
\newblock Domination number of some graphs.
\newblock 2018.  ISSN: 2455-2631.

\bibitem{rajasingh2012embedding}
Rajasingh I, Rajan B, Rajan RS.
\newblock Embedding of hypercubes into necklace, windmill and snake graphs.
\newblock \emph{Information Processing Letters}, 2012.
\newblock \textbf{112}(12):509--515.
doi:10.1016/j.ipl.2012.03.006.

\bibitem{DBLP:conf/mfcs/BezrukovCHRS98}
Bezrukov SL, Chavez JD, Harper LH, R{\"{o}}ttger M, Schroeder U.
\newblock Embedding of Hypercubes into Grids.
\newblock In: Mathematical Foundations of Computer Science 1998, 23rd
  International Symposium, MFCS'98, Brno, Czech Republic, August 24-28, 1998,
  Proceedings, volume 1450 of \emph{Lecture Notes in Computer Science}.
  Springer, 1998 pp. 693--701.

\bibitem{DBLP:journals/dam/ManuelRRM09}
Manuel PD, Rajasingh I, Rajan B, Mercy H.
\newblock Exact wirelength of hypercubes on a grid.
\newblock \emph{Discret. Appl. Math.}, 2009.
\newblock \textbf{157}(7):1486--1495.

\bibitem{DBLP:journals/dm/BezrukovCHRS00}
Bezrukov SL, Chavez JD, Harper LH, R{\"{o}}ttger M, Schroeder U.
\newblock The congestion of \emph{n}-cube layout on a rectangular grid.
\newblock \emph{Discret. Math.}, 2000.
\newblock \textbf{213}(1-3):13--19.

\bibitem{DBLP:journals/dm/Carlson02}
Carlson TA.
\newblock The edge-isoperimetric problem for discrete tori.
\newblock \emph{Discret. Math.}, 2002.
\newblock \textbf{254}(1-3):33--49.

\bibitem{rottger2001efficient}
R{\"o}ttger M, Schroeder UP.
\newblock Efficient embeddings of grids into grids.
\newblock \emph{Discrete Applied Mathematics}, 2001.
\newblock \textbf{108}(1-2):143--173.

\bibitem{opatrny2000embeddings}
Opatrny J, Sotteau D.
\newblock Embeddings of complete binary trees into grids and extended grids
  with total vertex-congestion 1.
\newblock \emph{Discrete Applied Mathematics}, 2000.
\newblock \textbf{98}(3):237--254.

\bibitem{bezrukov2001embedding}
Bezrukov SL.
\newblock Embedding complete trees into the hypercube.
\newblock \emph{Discrete Applied Mathematics}, 2001.
\newblock \textbf{110}(2-3):101--119.

\bibitem{caha2001optimal}
Caha R, Koubek V.
\newblock Optimal embeddings of generalized ladders into hypercubes.
\newblock \emph{Discrete Mathematics}, 2001.
\newblock \textbf{233}(1-3):65--83.

\bibitem{er1986lexicographical}
Er M.
\newblock Lexicographical Ordering of k-Subsets of a Set.
\newblock \emph{Journal of Information and Optimization Sciences}, 1986.
\newblock \textbf{7}(2):113--116.

\bibitem{DBLP:journals/siamdm/BollobasL90}
Bollob{\'{a}}s B, Leader I.
\newblock An Isoperimetric Inequality on the Discrete Torus.
\newblock \emph{{SIAM} J. Discret. Math.}, 1990.
\newblock \textbf{3}(1):32--37.

\bibitem{DBLP:journals/tcs/LaiT10}
Lai P, Tsai C.
\newblock Embedding of tori and grids into twisted cubes.
\newblock \emph{Theor. Comput. Sci.}, 2010.
\newblock \textbf{411}(40-42):3763--3773.

\bibitem{DBLP:journals/isci/FanJ07}
Fan J, Jia X.
\newblock Embedding meshes into crossed cubes.
\newblock \emph{Inf. Sci.}, 2007.
\newblock \textbf{177}(15):3151--3160.

\bibitem{DBLP:journals/cj/MillerRPR15}
Miller M, Rajan RS, Parthiban N, Rajasingh I.
\newblock Minimum Linear Arrangement of Incomplete Hypercubes.
\newblock \emph{Comput. J.}, 2015.
\newblock \textbf{58}(2):331--337.

\bibitem{DBLP:journals/tjs/BoalsGS94}
Boals AJ, Gupta AK, Sherwani NA.
\newblock Incomplete hypercubes: Algorithms and embeddings.
\newblock \emph{J. Supercomput.}, 1994.
\newblock \textbf{8}(3):263--294.

\bibitem{lu2011link}
L{\"u} L, Zhou T.
\newblock Link prediction in complex networks: A survey.
\newblock \emph{Physica A: statistical mechanics and its applications}, 2011.
\newblock \textbf{390}(6):1150--1170.

\bibitem{DBLP:books/fm/GareyJ79}
Garey MR, Johnson DS.
\newblock Computers and Intractability: {A} Guide to the Theory of
  NP-Completeness.
\newblock W. H. Freeman, 1979.
\newblock ISBN 0-7167-1044-7.

\bibitem{harperglobal}
Harper L, Chavez J.
\newblock Global methods of combinatorial optimization.
\newblock \emph{Preprint, Cambridge University Press}.

\bibitem{harper1964optimal}
Harper LH.
\newblock Optimal assignments of numbers to vertices.
\newblock \emph{Journal of the Society for Industrial and Applied Mathematics},
  1964.
\newblock \textbf{12}(1):131--135.

\bibitem{hamdi1994embedding}
Hamdi M.
\newblock Embedding hierarchical networks into the hypercube.
\newblock In: Proceedings of 1994 37th Midwest Symposium on Circuits and
  Systems, volume~1. IEEE, 1994 pp. 302--305.

\bibitem{choudum2009embedding}
Choudum SA, Indhumathi R.
\newblock On embedding subclasses of height-balanced trees in hypercubes.
\newblock \emph{Information Sciences}, 2009.
\newblock \textbf{179}(9):1333--1347.

\bibitem{bezrukov1999edge}
Bezrukov SL.
\newblock Edge isoperimetric problems on graphs.
\newblock \emph{Graph Theory and Combinatorial Biology}, 1999.
\newblock \textbf{7}:157--197.

\bibitem{prasanna}
L PN.
\newblock Applications of Graph Labeling in Communication Networks.
\newblock \emph{Orient. J. Comp. Sci. and Technol}, 2014.
\newblock \textbf{7}(1).

\bibitem{10.1117/12.401809}
Arkut IC, Arkut RC, Ghani N.
\newblock {Graceful label numbering in optical MPLS networks}.
\newblock In: Chlamtac I (ed.), OptiComm 2000: Optical Networking and
  Communications, volume 4233. 2000 pp. 1 -- 8.

\bibitem{article}
Goyal P, Ferrara E.
\newblock Graph Embedding Techniques, Applications, and Performance: A Survey.
\newblock \emph{Knowledge-Based Systems}, 2017,
doi:10.1016/j.knosys.2018.03.022.

\bibitem{leighton2014introduction}
Leighton FT.
\newblock Introduction to parallel algorithms and architectures:
  Arrays{\textperiodcentered} trees{\textperiodcentered} hypercubes.
\newblock Elsevier, 2014.

\bibitem{jung2008constructing}
Jung-hyun S, HyeongOk L, Moon-suk J.
\newblock Constructing complete binary trees on Petersen-torus networks.
\newblock In: 2008 Third International Conference on Convergence and Hybrid
  Information Technology, volume~2. IEEE, 2008 pp. 252--255.
  doi:10.1109/ICCIT.2008.54.

\bibitem{ramya2017labeling}
Ramya M, Meenakshi S.
\newblock Labeling Techniques of Some Special Graphs.
\newblock \emph{International Journal of Pure and Applied Mathematics}, 2017.
\newblock \textbf{116}(24):93--102.  ISSN:1311-8080.

\bibitem{day1997fault}
Day K, Al-Ayyoub AE.
\newblock Fault diameter of k-ary n-cube networks.
\newblock \emph{IEEE Transactions on Parallel and Distributed Systems}, 1997.
\newblock \textbf{8}(9):903--907.  doi:10.1109/71.615436.

\bibitem{day2000minimal}
Day K, Al-Ayyoub AE.
\newblock Minimal fault diameter for highly resilient product networks.
\newblock \emph{IEEE Transactions on Parallel and Distributed Systems}, 2000.
\newblock \textbf{11}(9):926--930.  doi:10.1109/71.879775.

\bibitem{bae2003edge}
Bae MM, Bose B.
\newblock Edge disjoint Hamiltonian cycles in k-ary n-cubes and hypercubes.
\newblock \emph{IEEE Transactions on Computers}, 2003.
\newblock \textbf{52}(10):1271--1284.
doi:10.1109/TC.2003.1234525.

\bibitem{wang2014novacube}
Wang T, Su Z, Xia Y, Qin B, Hamdi M.
\newblock NovaCube: A low latency Torus-based network architecture for data
  centers.
\newblock In: 2014 IEEE Global Communications Conference. IEEE, 2014 pp.
  2252--2257.  doi:10.1109/GLOCOM.2014.7037143.
\end{thebibliography}
\end{document}